\documentclass[11pt]{article}
\usepackage{amssymb,amsmath}
\usepackage{jeffe}
\usepackage{graphicx}
\usepackage{fullpage}
\usepackage{subfigure}
\usepackage{graphics}
\usepackage{color}
\usepackage{hyperref}

\newtheorem{theorem}{Theorem}[section]
\newtheorem{lemma}[theorem]{Lemma}

\newtheorem{corollary}[theorem]{Corollary}
\newtheorem{definition}[theorem]{Definition}

\newcommand{\OPT}{\mathsf{OPT}}

\renewcommand{\eg}{\mathsf{eg}}

\newcommand{\dG}{\vec{G}}
\newcommand{\uG}{G}

\renewcommand{\phi}{\varphi}

\title{A near-optimal approximation algorithm for Asymmetric TSP on embedded graphs}
\author{
Jeff Erickson\thanks{Department of Computer Science, University
of Illinois at Urbana-Champaign; \href{http://www.cs.uiuc.edu/~jeffe};  \url{jeffe@illinois.edu}.  Supported in part by NSF grant CCF-0915519.} 
\and
Anastasios Sidiropoulos\thanks{Department of Computer Science, University
of Illinois at Urbana-Champaign; \url{http://www.sidiropoulos.org/}; \url{sidiropo@illinois.edu}.  Supported in part by David and Lucille Packard Fellowship and by NSF grants CCF-0915984 and CCF-0915519.}
}

\date{}

\begin{document}

\setcounter{page}{0}
\maketitle

\begin{abstract}
We present a near-optimal polynomial-time approximation algorithm for the asymmetric traveling salesman problem for graphs of bounded orientable or non-orientable genus.  Our algorithm achieves an approximation factor of $O(f(g))$ on graphs with genus $g$, where $f(n)$ is the best approximation factor achievable in polynomial time on arbitrary $n$-vertex graphs.  In particular, the $O(\log n / \log \log n)$-approximation algorithm for general graphs by Asadpour~\etal\ [SODA~2010] %this got the best-paper award in SODA, not STOC
immediately implies an $O(\log g / \log \log g)$-approximation algorithm for genus-$g$ graphs.  Our result improves the $O(\sqrt{g} \log g)$-approximation algorithm of Oveis Grahan and Saberi [SODA 2011], which applies only to graphs with \emph{orientable} genus $g$.
We remark that our result gives the first algorithm that works even in the non-orientable case.

Moreover, using recent progress on the problem of approximating the genus of a graph, our $O(\log g / \log\log g)$-approximation can be implemented even without an embedding, for the case of bounded-degree graphs.
In contrast, obtaining a $O(\sqrt{g} \log g)$-approximation via the algorithm of Oveis Grahan and Saberi requires a genus-$g$ embedding as part of the input.

Finally, our techniques lead to a $O(1)$-approximation algorithm for ATSP on graphs of genus $g$, with running time $2^{O(g)} \cdot n^{O(1)}$.
\end{abstract}

\thispagestyle{empty}
\newpage

\section{Introduction}

The Asymmetric Traveling Salesman Problem is one of the most fundamental, and well studied problems in combinatorial optimization.
An instance of ATSP consists of a pair $(\dG,c)$, where $\dG=(V,A)$ is a directed graph, and $c$ is a (not necessarily symmetric) cost function $c:A \to \mathbb{R}^{+}$.
We may assume w.l.o.g.~that the cost of any arc $a\in A$ is equal to the shortest-path distance between its endpoints.
The goal is to find a spanning closed walk of minimum cost.

Recently, Asadpour \etal~\cite{DBLP:conf/soda/AsadpourGMGS10} gave an $O(\log n / \log \log n)$-approximation algorithm for ATSP, improving the nearly 30 year old $O(\log n)$-approximating of Frieze \etal~\cite{DBLP:journals/networks/FriezeGM82}. Building on this breakthrough, Oveis Grahan, and Saberi \cite{DBLP:conf/soda/GharanS11} gave a $O(\sqrt{g} \log g)$-approximation for graphs of orientable genus $g$.
We refer the reader to \cite{DBLP:conf/soda/AsadpourGMGS10, DBLP:conf/soda/GharanS11}, and the references therein, for a more detailed overview of the rich history of the problem.

We obtain an improved approximation algorithm for graph of genus $g$, by generalizing the upper bound of \cite{DBLP:conf/soda/AsadpourGMGS10} for general graphs. More precisely, we show that an $f(n)$-approximation for ATSP on general $n$-vertex graphs, implies an $O(f(g))$-aproximation for ATSP on graphs of either orientable, or non-orientable genus $g$.
The following summarizes our main result.

\begin{theorem}\label{thm:main}
If there exists a polynomial-time $f(n)$-approximation for ATSP on $n$-vertex graphs, then there exists polynomial-time $O(f(g))$-approximation for ATSP on graphs embedded into a surface of either orientable, or non-orientable genus $g$.
\end{theorem}

Combined with the $O(\log n / \log \log n)$-approximation of \cite{DBLP:conf/soda/AsadpourGMGS10} for $n$-vertex graphs, we immediately obtain the following corollary.

\begin{theorem}\label{thm:loggloglogg}
There exists a polynomial-time $O(\log g / \log \log g)$-approximation algorithm for ATSP, on graphs embedded into a surface of either orientable, or non-orientable genus $g$.
\end{theorem}

We remark that the result from \cite{DBLP:conf/soda/GharanS11} works only on the orientable case, while our result works on both orientable, and non-orientable surfaces.
Our algorithm implies therefore the first constant-factor approximation for graphs of bounded non-orientable genus.

Recently, Chekuri and Sidiropoulos \cite{chandra_sid_genus} have obtained a polynomial-time algorithm which given a graph $G$ of bounded maximum degree, outputs an embedding of $G$ into a surface of orientable (resp.~non-orientable) genus $g^{\alpha}$, where $g$ is the orientable (resp.~non-orientable) genus of $G$.
Combining this result with Theorem \ref{thm:main}, we immediately obtain an $O(\log g / \log \log g)$-approximation algorithm for ATSP on bounded-degree graphs of genus $g$, even when a embedding of the input graph is not given as part of the input.
This is summarized in the following.

\begin{theorem}\label{thm:main_no_embedding}
There exists a polynomial-time $O(\log g / \log \log g)$-approximation algorithm for ATSP, on graphs of either orientable, or non-orientable genus $g$, even when an embedding of the graph is not given as part of the input.
\end{theorem}

In contrast, if we apply the same approach to the $O(\sqrt{g} \log g)$-approximation algorithm of Oveis Grahan and Saberi, the resulting approximation guarantee deteriorates to $O(g^{\beta})$, for some constant $\beta > 6$.

Finally, our approximation guarantee can be improved to a constant, in time exponential in the genus of the input graph.

\begin{theorem}\label{thm:FPT}
There exists a $O(1)$-approximation algorithm for ATSP, on graphs embeddable into a surface of either orientable, or non-orientable genus $g$, with running time $2^{O(g)} \cdot n^{O(1)}$.
Moreover, this algorithm does not require a drawing of the graph as part of the input.
\end{theorem}

In other words, the above theorem shows that obtaining a constant-factor approximation for ATSP if fixed-parameter tractable, when parameterized by the genus of the input graph.
We note that since the running time of our algorithm is $2^{O(g)} \cdot n^{O(1)}$, we may assume that a drawing of the input graph into a surface of (either orientable, or non-orientable) genus $g$ is given.
This is because it has been shown by Kawarabayashi, Mohar, and Reed \cite{KawarabayashiMR08} that such a drawing can be found in time $2^{O(g)} \cdot n$.
Our result implies in particular the following. 

\begin{corollary}\label{thm:FPT}
There exists a polynomiam-time $O(1)$-approximation algorithm for ATSP, on graphs embeddable into a surface of either orientable, or non-orientable genus $O(\log n)$.
Moreover, this algorithm does not require a drawing of the graph as part of the input.
\end{corollary}

\subsection{A high-level description of the algorithm}
\paragraph{Previous work on rounding the Held-Karp LP.}
Our algorithm uses the Held-Karp LP relaxation of ATSP, and 
it is inspired by, and uses ideas from \cite{DBLP:conf/soda/AsadpourGMGS10,DBLP:conf/soda/GharanS11}.
We begin our description by recalling some key steps from these algorithms.
%Both of the algorithms in \cite{DBLP:conf/soda/AsadpourGMGS10, DBLP:conf/soda/GharanS11} are based on the same relaxation.
Intuitively, a feasible solution to the Held-Karp LP assigns a weight to every arc, so that the total weight crossing every cut is at least 1 (in every direction), and the total weight entering/leaving every vertex is 1 (see Section \ref{sec:LP} for a formal definition).
The main tool used in \cite{DBLP:conf/soda/AsadpourGMGS10,DBLP:conf/soda/GharanS11} for rounding such a fractional solution is the notion of a \emph{thin spanning tree}.
Roughly speaking, a tree $T$ is $(\alpha,s)$-thin, for some $\alpha,s>0$, if (i) for every cut, the total number of edges in $T$ crossing the cut is at most $\alpha$ times the total fractional cost of the cut, and (ii) the total cost of the tree is at most $s$ times the cost of the fractional solution.
Given such an $(\alpha,s)$-thin spanning tree, one can obtain a tour of total cost $O((\alpha+s) \cdot \OPT)$, where $\OPT$ is the cost of the minimum TSP tour, via a careful application of Hoffman's circulation theorem.

\paragraph{The forest for the trees.}
We depart from the above paradigm by using \emph{thin spanning forests} instead of thin spanning trees.
Given a graph of genus $g$, we show how to construct a $(O(1), O(1))$-thin spanning forest with at most $g$ connected components.
Using a modified application of Hoffman's circulation theorem, we can transform this thin forest into a collection of $g$ closed walks $W_1,\ldots,W_g$, visiting all vertices in the graph, and having total cost $O(\OPT)$.
We then apply an idea from \cite{DBLP:journals/networks/FriezeGM82} to ``merge'' these walks into a single one.
More specifically, we pick an arbitrary vertex from every walk, and we form a smaller ATSP instance on a graph with $g$ vertices.
We solve this smaller instance using the algorithm for general graphs, obtaining a tour $C$.
Finally, we produce the solution to our problem by merging $C$ with $W_1,\ldots,W_g$, and shortcutting the resulting walk.

\paragraph{How can we find a thin forest?}
The main technical part of our algorithm is the computation of a thin spanning forest with a small number of connected components.
This is done by introducing a certain structure that we call a \emph{ribbon}.
Intuitively, this is a maximal set of parallel edges that are contained inside a disk in the surface.
We show that unless the graph has at most $g$ vertices, we can find such a ribbon having large total fractional cost.
We then \emph{contract} all the edges in this ribbon, and repeat until we arrive at a graph with at most $g$ vertices.
Every vertex in the contracted graph corresponds to a connected component of ribbons in the original graph.
The fact that every ribbon has large fractional cost allows us to find a spanning tree in every such component, so that the resulting spanning forest is thin.

\subsection{Preliminaries}

%\paragraph{Surfaces.}
We give a brief overview of some of the notation that we will use.
For a more detailed exposition we refer the reader to  \cite{mohar2001graphs}.
A \emph{surface} is a compact 2-dimensional manifold; it is called \emph{orientable} if it is embeddable in $\mathbb{R}^3$, and \emph{non-orientable} otherwise.
An \emph{embedding} of an undirected graph $G$ in a surface ${\cal S}$ is a continuous injective function from the graph (as a topological space) to ${\cal S}$.  Vertices of $G$ are mapped to distinct points in ${\cal S}$; edges are mapped to simple, interior-disjoint paths.  A face of an embedding is a component of the complement of the image of the embedding.  Without loss of generality, we consider only cellular embeddings, meaning every face is homeomorphic to an open disk.  To avoid excessive notation, we do not distinguish between vertices, edges, and subgraphs of $G$ and their images under the embedding.  A bigon is a face bounded by exactly two edges.

A cycle is \emph{contractible} if it can be continuously deformed to a point; classical results of Epstein \cite{e-c2mi-66} imply that a simple cycle in ${\cal S}$ is contractible if and only if it is the boundary of a disk in ${\cal S}$.  Two paths with matching endpoints are \emph{homotopic} if they form a contractible cycle.  Thus, two edges in an embedded graph are homotopic if and only if they bound a bigon.

The \emph{dual} $G^*$ of an embedded graph $G$ is defined as follows.  For each face $f$ of $G$, the dual graph has a corresponding vertex $f^*$.  For each edge $e$ of $G$, the dual graph has an edge $e^*$ connecting the vertices dual to the two faces on either side of $e$.  Intuitively, one can think of $f^*$ as an arbitrary point in the interior of $f$ and $e^*$ has a path that crosses $e$ at its midpoint and does not intersect any other edge of $G$.  Each face of $G^*$ corresponds to a vertex of $G$; thus, the dual of $G^*$ is isomorphic to the original embedded graph $G$.  Trivially, a face $f$ of $G$ is a bigon if and only if its dual vertex $f^*$ has degree $2$.

The \emph{genus} of a surface ${\cal S}$ is the maximum number of cycles in ${\cal S}$ whose complement is connected.
Let $G$ be an embedded graph with $n$ vertices, $m$ edges, and $f$ faces.  \emph{Euler's formula} states that $n-m+f = 2-\chi({\cal S})$, where $\chi({\cal S})$ is a topological invariant of the surface called its \emph{Euler characteristic}.  For a surface of genus $g$, the Euler characteristic is $2-2g$ if the surface is orientable and $2-g$ if the surface is non-orientable.  To simplify our notation, we define the \emph{Euler genus} of ${\cal S}$ as $\eg({\cal S}) := 2-\chi({\cal S})$.

\iffalse
Let ${\cal S}$ be a surface.
For any triangulation of ${\cal S}$, with $v$ vertices, $e$ edges, and $f$ faces, the quantity $v-e+f$ is a topological invariant, called
the \emph{Euler characteristic} of ${\cal S}$, and denoted by $\chi({\cal S})$.
The \emph{Euler genus} of ${\cal S}$ is defined to be $\eg({\cal S}) = 2-\chi({\cal S})$.
If ${\cal S}$ is orientable, then the \emph{orientable genus} of ${\cal S}$ is defined to be $g = 1-\chi({\cal S})/2$.
If ${\cal S}$ is non-orientable, then the \emph{non-orientable genus} of ${\cal S}$ is defined to be $k = 2-\chi({\cal S})$.
Therefore, for any orientable surface ${\cal S}$ of orientable genus $g$, we have $\eg({\cal S}) = 2g$, and for any non-orientable surface ${\cal S}$ of non-orientable genus $k$, we have $\eg({\cal S}) = k$.
In order to simplify the notation, and since this only affects our bounds by a constant factor, we will restrict the rest of the discussion to the Euler genus.

The Euler genus of a graph $G$, denoted by $\eg(G)$, is defined to be the minimum integer $\gamma\geq 0$, such that $G$ can be embedded into a surface of Euler genus $\gamma$.
Will often abuse notation, and treat graphs as topological subspaces of the surfaces into which they are embedded.
This will allow as to identify the vertices in a graph with points, and the edges with curves in the surface.
In particular, contracting a set of edges in the graph will correspond to contracting a set of curves in the surface.
\fi

\subsection{Organization}
The rest of the paper is organized as follows.
In Section \ref{sec:LP} we recall the Held-Karp LP relaxation of ATSP.
In Section \ref{sec:ribbon_decompositions} we introduce the notion of a ribbon, and prove some basic properties of decompositions of the edges of a graph into a collection of ribbons.
Using these ribbon decompositions, we show in Section \ref{sec:thin_forest} how to construct a thin forest, with a small number of connected components.
In Section \ref{sec:circulation} we how how to turn a thin forest with a small number connected components into a small collection of closed walks visiting all the vertices, and with small total cost.
Finally, in Section \ref{sec:algorithm} we show how to combine the above technical ingredients to obtain the algorithm for ATSP.

\section{The Held-Karp LP relaxation}\label{sec:LP}

We now recall the definition of the Help-Karp LP relaxation of ATSP.
Let $(\dG,c)$ be an instance of ATSP, where $\dG=(V,A)$ is a directed graph, and $c:A \to \mathbb{R}^+$.
For a set $U\subset V$, we define
\begin{align*}
\delta_{\dG}^+(U) &= \{u \rightarrow v \in A : u\in U, v\notin U\},\\
\delta_{\dG}^-(U) &= \delta_{\dG}^+(V\setminus U).
\end{align*}
For a vertex $v\in V$, we write $\delta_{\dG}^+(v) = \delta_{\dG}^+(\{v\})$,
and $\delta_{\dG}^-(v) = \delta_{\dG}^-(\{v\})$.
%We omit the subscripts when $\dG$ is clear from the context.
We define the undirected graph $\uG=(V,E)$, with
\[
E = \left\{uv \in {V \choose 2} : u\rightarrow v\in A, \mbox{ or } v\rightarrow u\in A\right\}.
\]
For any $U\subseteq V$, let
\[
\delta_{\uG}(U) = \{uv \in E : u\in U, v\notin U\}.
\]
We omit the subscript when $\uG$ is clear form the context.
We also extend $c$ to $E$ as follows.
For any $uv\in E$, we define
\[
c(uv) = \min\{ c(u\rightarrow v), c(v\rightarrow u) \}.
\]
For a function $x:A\to \mathbb{R}$, and for all $W\subseteq A$, we write $x(W) = \sum_{a\in W} x(a)$.
Using the above notation, the Held-Karp LP relaxation is defined as follows.

\iffalse
\begin{align*}
\min &\sum_{a\in A} c(a) \cdot x(a)\\
\mbox{s.t.~} &x(\delta^+(U)) \geq 1 &\text{for all } U\subseteq V\\
 &x(\delta^+(v)) = x(\delta^-(v)) = 1 &\text{for all } v\in V\\
 &x(a) \geq 0 &\text{for all } a\in A
\end{align*}
\fi

\begin{center}
\begin{tabular}{|lll|}
\hline
\multicolumn{3}{|l|}{Held-Karp LP}\\
\hline
$\min$ &$\sum_{a\in A} c(a) \cdot x(a)$ & \\
$\mbox{s.t.~}$ &$x(\delta^+(U)) \geq 1$ &$\text{for all } U\subseteq V$\\
 &$x(\delta^+(v)) = x(\delta^-(v))$ &$\text{for all } v\in V$\\
 &$x(a) \geq 0$ &$\text{for all } a\in A$\\
 \hline
\end{tabular}
\end{center}
We remark that typically the Held-Karp LP is defined on a complete graph, with the additional constrain $x(\delta^+(v)) = 1$, for all $v\in V$.
However, we will use the above (possibly weaker) formulation, which turns out to be sufficient for our analysis.
The benefit of doing so is that any feasible solution  is supported on the arcs of the input graph $G$.
%This is without loss of generality, since in we instead define the 

Let $x:A\to \mathbb{R}$ be a feasible solution for the Held-Karp LP.
We define the \emph{symmetrization} of $x$ to be a function $z:E \to \mathbb{R}$, such that for any $uv \in E$, we have
\[
z(uv) = x(u\rightarrow v) + x(v \rightarrow u).
\]

\iffalse
\begin{lemma}[Asadpour \etal~\cite{DBLP:conf/soda/AsadpourGMGS10}]
Let $(G,c)$ be an instance of ATSP.
Let $x$ be a feasible solution for the Held-Karp relaxation of $(G,c)$.
Let $\alpha,s>0$.
Let $T$ be a spanning subtree of $\overline{G}$, such that $T$ is $(\alpha,s)$-thin w.r.t.~$x$.
Then, there exists a polynomial time algorithm which given $G,c,x$, and $T$, outputs a solution ATSP, with cost at most
\[
(2\alpha + s)\cdot \sum_{a\in A} c(a) \cdot x(a).
\]
\end{lemma}

\begin{lemma}[Asadpour \etal~\cite{DBLP:conf/soda/AsadpourGMGS10}]
There exists a randomized polynomial time algorithm which given an instance $(G,c)$ of ATSP, and a feasible solution $x$ for the Held-Karp   relaxation of $(G,c)$, computes a $\left(O(\frac{\log |V|}{\log\log |V|}), 2\right)$-thin (w.r.t.~$x$) spanning subtree $T$ of $\overline{G}$,
with high probability.
\end{lemma}
\fi

\section{Ribbon decompositions}\label{sec:ribbon_decompositions}

In this section we introduce some of the machinery that we will use in our algorithm for computing a thin forest.

Let $G$ be a graph embedded into a surface.
A \emph{ribbon} in $G$ is a maximal set $R$ of parallel non-self-loop edges in $G$, such that for any $e \neq e'\in R$, the cycle $e \cup e'$ is contractible (see Figure \ref{fig:ribbon} for an example).
Note that for every ribbon $R$, if $|R|>1$, then there exists a disk in the surface that obtains $R$.
The edges in $R$ are naturally ordered inside this disk, so that every two consecutive edges in $R$ bound a face.
We say that an edge in $R$ is \emph{central} if it is a weighted (w.r.t.~$z$) median in this ordering.
Note that every ribbon has either one, or two central edges.

\begin{figure}
\begin{center}
\scalebox{0.7}{\includegraphics{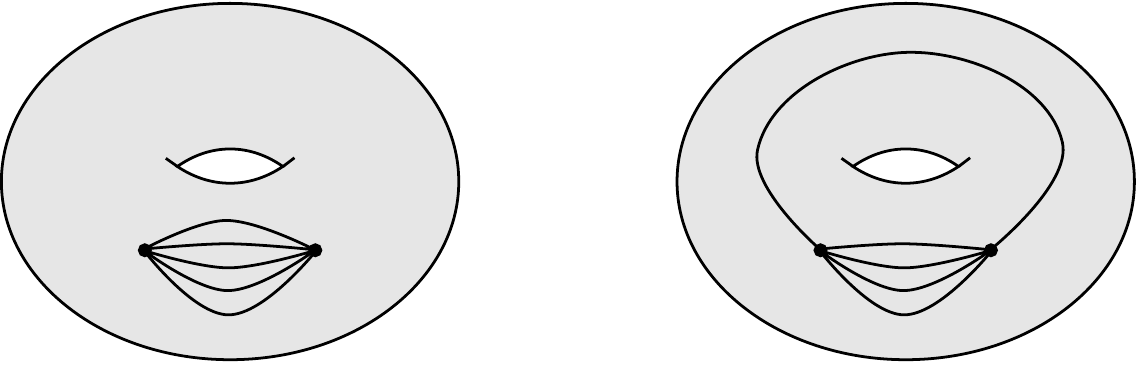}}
\caption{An example of a set of parallel edges forming a ribbon (left), and a set of parallel edges \emph{not} forming a ribbon (right).\label{fig:ribbon}}
\end{center}
\end{figure}

\begin{lemma}\label{lem:ribbon_decomposition}
Let $G$ be a graph embedded into some surface.
Then, the set of non-self-loop edges of $G$ can be uniquely decomposed into a collection of pairwise disjoint ribbons.
\end{lemma}
\begin{proof}
It follows immediately from the definition of a ribbon.
\end{proof}

\begin{lemma}\label{lem:edges_Euler}
Let $G=(V,E)$ be a multi-graph embedded into a surface ${\cal S}$.
Let ${\cal R}$ be a decomposition of the non-self-loop edges of $G$ into ribbons.
Then, $|{\cal R}| \leq 3 |V| - 3 \chi({\cal S})$.
\end{lemma}
\begin{proof}
Let $Y\subseteq E$ be a set obtained by choosing an arbitrary edge from every ribbon $R \in {\cal R}$.
Let $G'=(V,Y)$ be the subgraph of $G$ induced by $Y$.
Since for every pair of edges $e,e'\in Y$ with common endpoints, we have that the loop $\phi(e)\cup \phi(e')$ is noncontractible, and there are no self-loops in $Y$, it follows that every face of $G'$ contains at least 3 distinct edges.
Let $F$ denote the set of faces of $G'$.
It follows that $3 |F| \leq 2 |Y|$.
By Euler's formula, we obtain
\begin{align*}
|{\cal R}| &= |Y|\\
 &\leq 3 |V| - 3 \chi({\cal S})\\
 &= 3 |V| - 3 \chi({\cal S}),
\end{align*}
as required.
\end{proof}

\section{Computing a thin forest}\label{sec:thin_forest}

As before, let $(\dG,c)$ be an instance of ATSP, with $\eg(\dG)=g$, let $x$ be a feasible solution to the Held-Karp relaxation of $(\dG,c)$, and let $z$ be the symmetrization of $x$.
The notion of a \emph{thin} set captures a key idea for rounding a solution of the Held-Karp LP (see \cite{DBLP:conf/soda/AsadpourGMGS10,DBLP:conf/soda/GharanS11}).
We begin by recalling the definition of a thin set.

\begin{definition}[Thinness]
Let $W\subseteq E$, and $\alpha,s>0$.
We say that $W$ is \emph{$(\alpha,s)$-thin} (w.r.t.~$x$) if 
for any $U\subseteq V$,
\[
|W\cap \delta(U)| \leq \alpha \cdot z(\delta(U)),
\]
and
\[
c(W) \leq s \cdot \sum_{a\in A} c(a) \cdot x(a).
\]
%For a subgraph $H \subseteq \uG$, we also say that $H$ is $(\alpha, s)$-thin, if $E(H)$ is $(\alpha,s)$-thin.
\end{definition}

In this section we show how to compute a $(O(1),O(1))$-thin forest, with at most $g$ connected components.

We begin by first showing how to compute a forest $T$ with at most $g$ components, satisfying a slightly weaker notion of thinness.
To that end, we compute a sequence of graphs $G_0,\ldots,G_t$, with $G_0=G$, as follows.
Let $i\geq 0$, and suppose we have computed $G_i$.
If $G_i$ has at most $g$ vertices, then we terminate the sequence at $G_i$, and we set $t=i$.
Otherwise, let ${\cal R}_i$ be the decomposition of the non-self-loop edges in $G_i$ into ribbons (by Lemma \ref{lem:ribbon_decomposition} the decomposition ${\cal R}_i$ is unique).
We set
\[
R_i = \argmax_{R' \in {\cal R}_i} z(R').
\]
We let $G_{i+1}$ be the graph obtained from $G_i$ by contracting all the edges in $R_i$.
%Notice that since $R_i$ is a ribbon in $G_i$, all the edges in $R_i$ are contained inside a disk, and contracting $R_i$ corresponds to contracting this disk.
This concludes the definition of the sequence of graphs $G_0,\ldots,G_t$ (see Figure \ref{fig:ribbon_contraction_sequence} for an example).
For every $i\in \{0,\ldots,t-1\}$, let $e_i$ be a central edge in $R_i$.
We define $T$ to be the graph induced by $\{e_0,\ldots,e_{t-1}\}$.
Since ribbons do not contain self-loop edges, and we only contract ribbons, it immediately follows that $T$ is a forest.
It is also immediate that $T$ contains one connected component for every vertex of $G_t$, and therefore it contains at most $g$ connected components.

It remains to show that $T$ satisfies a weak thinness condition.
We first derive the following auxiliary fact.

\begin{figure}
\begin{center}
\scalebox{0.7}{\includegraphics{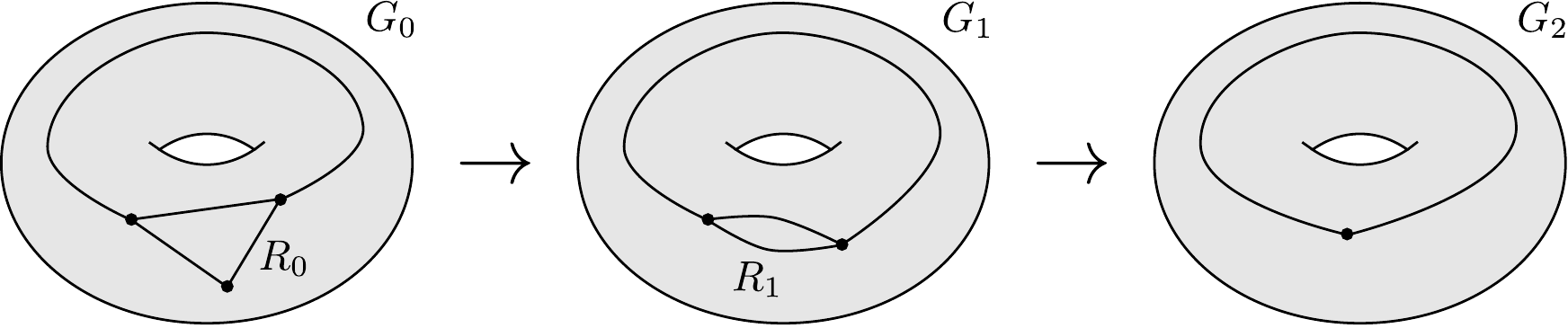}}
\caption{An example of a sequence of ribbon contractions preformed by our algorithm.\label{fig:ribbon_contraction_sequence}}
\end{center}
\end{figure}

\begin{lemma}\label{lem:contraction_sequence}
For any $i\in \{0,\ldots,t-1\}$, we have $z(R_i) \geq 2/5$.
\end{lemma}
\begin{proof}
Let $i\in \{0,\ldots,t-1\}$, $G_i = (V_i, E_i)$, and  $U\subseteq V_i$.
Since $G_i$ is obtained from $G$ via a sequence of edge contractions, it follows that there exists $U'\subseteq V$, such that $\delta_{G_i}(U) = \delta_{\uG}(U')$. 
Thus, 
\begin{align}
z(\delta_{G_i}(U)) &= z(\delta_{\uG}(U')) \notag \\
 &= x(\delta_{\dG}^+(U')) + x(\delta_{\dG}^-(U')) \notag \\
 &\geq 2. \label{eq:boundaries}
\end{align}

By construction we have that for any $i\in \{0,\ldots,t-1\}$, $|V_i| \geq g$.
By lemma \ref{lem:edges_Euler} we have that for any $i\in \{0,\ldots,t-1\}$,
\begin{align*}
|{\cal R}_i| &\leq 3 |V_i| - 3\chi({\cal S})\\
 &= 3 |V_i| - 6 + 3 g \\
 &< 6 |V_i| - 6.
\end{align*}
Therefore, there exists $v_i\in V_i$, such that all the non-self-loop edges incident to $v_i$ are contained in at most $5$ ribbons.
By \eqref{eq:boundaries}, we have that
\[
\sum_{e\in \delta(\{v_i\})} z(e) \geq 2.
\]
This implies that there exists a ribbon $R_i'\in {\cal R}_i$, such that $z(R_i') \geq 2/5$.
Therefore,
\begin{align*}
z(R_i) &\geq z(R_i') \geq 2/5,
\end{align*}
which concludes the proof.
\end{proof}

We can now show that $T$ satisfies a weak notion of thinness.

\begin{lemma}\label{lem:thin1}
For any $U\subseteq V$, we have 
$|T \cap \delta(U)| \leq O(1) \cdot z(\delta(U))$.
\end{lemma}
\begin{proof}
For any $i\in \{0,\ldots,t\}$, let $G^*_i = (V_i^*, E_i^*)$ be the dual of $G_i$.
For any $i\in \{0,\ldots,t-1\}$, let $R_i^*\subseteq E_i^*$ be the set of the duals of all edges in $R_i$, and let $e_i^*$ be the dual of $e_i$.

Note that for any $i\in \{0,\ldots,t-1\}$, all self-loop edges in $G_i$ are noncontractible.
This is because self-loops can only be created when contracting the edges in a ribbon $R$.
For any such self-loop $e$, since $e\notin R$, it follows that there exists $e'\in R$, such that the loop $e\cup e'$ is noncontractible, which implies that $e$ becomes a noncontractible loop after contracting all edges in $R$.
This implies that for any $i\in \{0,\ldots,t-1\}$, there exists a path $P_i^*$ in $G_i^*$, such that $E(P_i^*) = R_i^*$, and all internal vertices in $P_i^*$ have degree $2$.

For any $i\in \{0,\ldots,t-1\}$, $G_{i+1}$ is obtained from $G_i$ by contracting all edges in $R_i$.
Since the edges in $R_i$ are non-self-loops,
this means that $G_{i+1}^*$ is obtained from $G_i^*$ by deleting all edges in $R_i^*$, and all internal vertices in $P_i^*$.

\begin{center}
\scalebox{0.7}{\includegraphics{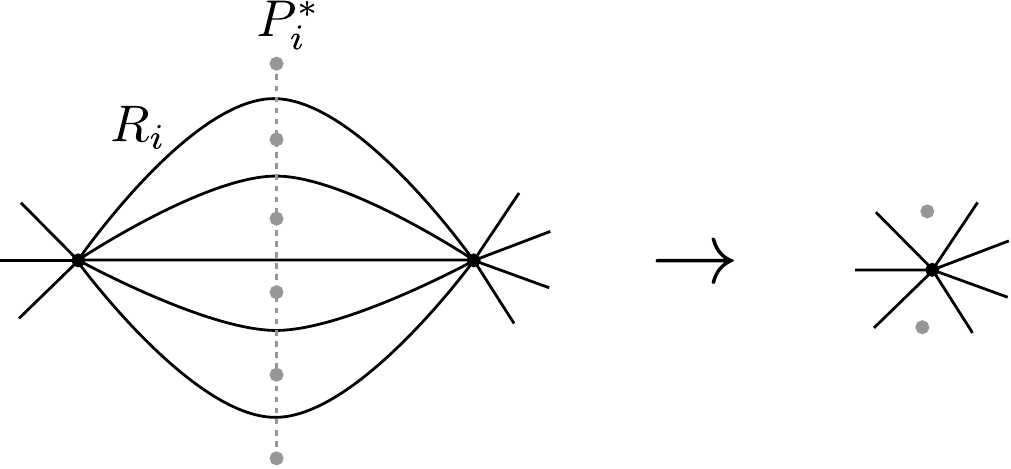}}
\end{center}

Consider some set $U\subseteq V$.
Let $X = \delta(U)$, and let $X^*$ be the set of the duals of all edges in $X$, i.e.~$X^* = \{e^* \in E^* : e\in X\}$.
There exists a collection ${\cal K}^*$ of pairwise edge-disjoint cycles  in $G^*$, such that
\[
X^* = \bigcup_{K^*\in {\cal K}^*} E(K^*).
\]

Let $K^*\in {\cal K}^*$.
Suppose fist that $K$, contains exactly one edge of $T$, and 
let $e_i$ be that edge.
It follows by construction that $P_i^* \subseteq K^*$, which implies
$z(K) \geq z(R_i) \geq 2/5$.
Otherwise, suppose that $K$, contains at least two edges of $T$.
Let $e_i^*$, $e_j^*$ be two such edges that are consecutive in $K^*$.
Assume w.l.o.g.~that $i<j$.
By construction we have that the subpath $Q^*$ of $K^*$ between (and including) $e_i^*$, and $e_j^*$ satisfies $z(Q) \geq z(R_i)/2 \geq 1/10$.
To summarize, we have that for any $K^*\in {\cal K}^*$
\begin{align*}
z(K) &\geq \frac{1}{20} |T \cap K|.
\end{align*}
Summing up over all cycles in ${\cal K}^*$, we obtain
\begin{align*}
|T \cap \delta(U)| &= \sum_{K^* \in {\cal K}^*} |T \cap K|\\
 &\leq \sum_{K^* \in {\cal K}^*} 20\cdot z(K)\\
 &= 20\cdot z(\delta(U)),
\end{align*}
concluding the proof.
\end{proof}

We are now ready to prove the main result of this section.
Lemma \ref{lem:thin1} gives a spanning forest that satisfies a weak notion of thinness.
More specifically, $T$ might have very large cost (w.r.t.~$c$).
We show how to transform this forest into a $(O(1),O(1))$-thin spanning forest.
Our proof uses the argument from Oveis Grahan, and Saberi \cite{DBLP:conf/soda/GharanS11}, who obtained this transformation for the case of spanning trees.
We observe that their proof works in the more general setting of spanning forests with a fixed number of connected components.
We give the proof for completeness.

\iffalse
\begin{lemma}[Oveis Grahan, and Saberi \cite{DBLP:conf/soda/GharanS11}]\label{lem:thin_cost}
Let $k>0$, and let ${\cal F}$ be the family of all spanning forests in $\uG$, with $k$ connected components.
Let $\alpha>0$, and 
suppose that there exists a polynomial-time algorithm which computes a forest $T\in {\cal F}$, such that for any $U\subseteq V$,
\[
|T\cap \delta(U)| \leq \alpha \cdot z(\delta(U)).
\]
Then, there exists a polynomial-time algorithm which computes a forest $T'\in {\cal F}$, such that $T'$ is $(2\alpha, 2\alpha)$-thin, w.r.t.~$x$.
\end{lemma}
\fi

%We are now ready to prove the main result of this section.

\begin{lemma}[Computing a thin forest]\label{lem:thin2}
There exists a polynomial time algorithm which 
computes a $(O(1), O(1))$-thin (w.r.t.~$x$) spanning forrest $T$ in $\uG$, with at most $g$ connected components.
\end{lemma}
\begin{proof}
We first remark that the algorithm of Lemma \ref{lem:thin1} returns a forest $T$ satisfying for all $U\subseteq V$, 
$|T \cap \delta(U)| \leq \alpha  z(\delta(U))$, for some $\alpha=O(1)$, even if $z$ is not the symmetrization of some feasible solution to the Held-Karp LP.
The only requirement for the algorithm of Lemma \ref{lem:thin1} is that $z$ is non-negative, and for all $U'\subseteq V$, we have $z(\delta_G(U')) \geq 2$.

Let $N=n^2/\alpha$.
We compute a sequence of non-negative functions $z_0,\ldots,z_N$, and a sequence of spanning forests $T_1,\ldots,T_N$, each having at most $g$ connected components, as follows.
We set $z_0=3\lfloor z n^2 \rfloor / n^2$, and inductively maintain the invariant that for any $i\in \{1,\ldots,N\}$: 
\begin{align}
\text{for all } U'\subseteq V, ~~ z(\delta_G(U')) \geq 2 \label{eq:invariant}
\end{align}
It is immediate that $z_0$ is non-negative, and satisfies \eqref{eq:invariant}.
Given $z_i$, for some $i\in \{0,\ldots,N-1\}$, we compute $z_{i+1}$, $T_{i+1}$ as follows.
Using Lemma \ref{lem:thin1} we find a spanning forrest $T_i$ with at most $g$ connected components, such that 
\begin{align}
\text{for all } U\subseteq V, ~~  
|T_{i+1} \cap \delta_G(U)| \leq \alpha z_{i}(\delta_G(U)) \label{eq:weak_thinness}
\end{align}
For any $e\in E$, we set
\begin{align*}
z_{i+1}(e) = \left\{\begin{array}{ll}
z_i(e) - 1/n^2 & \text{ if } e\in T_{i+1}\\
z_i(e) & \text{ if } e \notin T_{i+1}
\end{array}\right.
\end{align*}
We need to verify that \eqref{eq:invariant} holds for $z_{i+1}$.
We have that for any $U' \subseteq V$,
\begin{align}
z_{i+1}(d_G(U')) &\geq z_i(\delta_G(U')) - \alpha/n^2 \label{eq:inv1}\\
 &\geq z_0(\delta_G(U')) - N \alpha/n^2 \notag\\
 &> 3 (z(\delta_G(U')) - 1) - N \alpha / n^2 \label{eq:inv2}\\
 &\geq 2 \label{eq:inv3},
\end{align}
where \eqref{eq:inv1} follows by \eqref{eq:weak_thinness}, 
\eqref{eq:inv2} by the fact that $z_0=3\lfloor z n^2 \rfloor / n^2$, and every cut contains less than $n^2$ edges,
and \eqref{eq:inv3} by the fact that $z$ is the symmetrization of a feasible solution to the Held-Karp LP.
Therefore, the inductive invariant \eqref{eq:invariant} is maintained.
Lastly, we need to verify that $z_{i+1}$ is non-negative.
Note that $z_0$ takes values that are  multiples of $1/n^2$, and we only decrease the values of $z_i$ by $1/n^2$, which implies that for all $i$, all values in $z_i$ are multiples of $1/n^2$.
An edge $e$ can appear in the forest $T_{i+1}$ only if $z_i(e)>0$.
Therefore, when setting $z_{i+1}(e)=z_i(e)$, the value $z_{i+1}(e)$ remains non-negative.

We now set the desired forest $T$ to be 
\[
T = \argmin_{T' \in \{T_1,\ldots,T_N\}}  c(T').
\]
It remains to verify that $T$ is $(O(1),O(1))$-thin.
Suppose that $T=T_i$, for some $i\in \{1,\ldots,N\}$.
By \eqref{eq:weak_thinness}, we have that for all $U \subseteq V$,
\begin{align}
|T_{i} \cap \delta_G(U)| &\leq \alpha  z_{i-1}(\delta_G(U)) \notag\\
 &\leq \alpha  z_0(\delta_G(U)) \notag \\
 &\leq 3 \alpha  z(\delta_G(U)) \label{eq:thin_forest_1}
\end{align}
For every $uv\in E$, we have $\sum_{uv\in T_i} 1/n^2 \leq z_0(uv) \leq 3z(uv)$,
and therefore $c(uv) \sum_{uv\in T_i} 1/n^2 \leq 3 z(uv) c(uv) \leq 3 x(u\to v) c(u \to v) + x(v \to u) c(v \to u)$.
Summing over all edges $uv \in E$, we obtain
\begin{align}
c(T) &\leq \frac{1}{N} \sum_{i=1}^N c(T_i) \notag \\
 &\leq 3 \frac{1}{N} n^2 \sum_{a\in A} x(a) c(a) \notag \\
 &= 3\alpha \sum_{a\in A} x(a) c(a) \label{eq:thin_forest_2}
\end{align}
Combining \eqref{eq:thin_forest_1} \& \eqref{eq:thin_forest_2} we obtain that $T$ is $(3\alpha, 3\alpha)$-thin, as required.
\end{proof}

%\section{Walking in the forest}\label{sec:circulation}
\section{From thin forests to walks}\label{sec:circulation}

In this section we present the last missing ingredient of our algorithm.
More specifically, we show how given a thin forest with a few connected components, we can construct a small number of closed walks that span the input graph, and having small total cost.
This step is similar to the argument in \cite{DBLP:conf/soda/AsadpourGMGS10}.
However, here we need a slightly modified version since we are dealing with a thin forest, instead of a thin tree, so we include the proof for completeness.
We will use the following result due to Hoffman \cite{schrijver2003combinatorial}.

\begin{theorem}[Hoffman's circulation theorem \cite{schrijver2003combinatorial}]\label{thm:hoffman_circulation}
Let $\dG=(V,A)$ be a directed graph, and let $l,u:A\to \mathbb{R}$.
Then, there exists a circulation $f$ with $l(a) \leq f(a) \leq u(a)$ for all $a\in A$, if and only if the following two conditions are satisfied:
\begin{description}
\item{(1)}
$l(a) \leq u(a)$ for all $a\in A$.
\item{(2)}
For all $U\subseteq V$, we have $l(\delta^-(U)) \leq u(\delta^+(U))$.
\end{description}
Furthermore, if $l$ and $u$ are integer-valued, then $f$ can be chosen to be integer-valued.
\end{theorem}

We are now ready to prove the main result of this section.

\begin{lemma}[From thin forests to Eulerian walks]\label{lem:circulation}
%Let $(\dG,c)$ be an instance of ATSP, and let $x:V \to \mathbb{R}$ be a feasible solution for the Held-Karp relaxation of $(\dG,c)$.
Let $\alpha,s>0$, and let $T$ be a $(\alpha,s)$-thin (w.r.t.~$x$) spanning forest in $\uG$ with at most $k$ connected components.
Then, there exists a polynomial-time algorithm which computes a collection of closed closed walks $C_1,\ldots,C_{k'}$ in $\dG$, for some $k'\leq k$, that visit all the vertices in $\dG$, and such that 
$\sum_{i=1}^k c(C_i) \leq (2\alpha + s) \sum_{a\in A} c(a) \cdot x(a)$.
\end{lemma}

\begin{proof}
Define a set $\vec{T} \subset A$ as follows.
For every $\{u,v\} \in T$, if $c((u,v)) \leq c((v,u))$, then we add $(u,v)$ to $\vec{T}$, and otherwise we add $(v,u)$ to $\vec{T}$.

Define functions $l,u:A \to \mathbb{R}^+$, such that for any $a\in A$,
\begin{align*}
l(a) &= \left\{ \begin{array}{ll}
1 & a\in \vec{T}\\
0 & a\notin \vec{T}
\end{array} \right. ,
\end{align*}
and
\begin{align*}
u(a) &= \left\{ \begin{array}{ll}
1 + 2 \alpha x(a) & a\in \vec{T}\\
2 \alpha x(a) & a\notin \vec{T}
\end{array} \right. .
\end{align*}

We argue that there exists a circulation $f:A \to \mathbb{R}^+$, such that for any $a\in A$, $l(a)\leq f(a) \leq u(a)$.
To that end, it suffices to verify conditions (1), and (2) of theorem \ref{thm:hoffman_circulation}.
For any $a\in A$ we have $l(a) = u(a) - 2\alpha x(a) \leq u(a)$, and therefore condition (1) is satisfied.
To verify condition (2), consider a set $U\subseteq V$.
Let $z$ be the symmetrization of $x$.
We have
\begin{align}
l(\delta^-(U)) &= |\vec{T} \cap \delta^-(U)| \label{eq:hoff1}\\
 &\leq |T\cap \delta(U)| \label{eq:hoff2}\\
 &\leq \alpha z(\delta(U)) \label{eq:hoff3}\\
 &= \alpha (x(\delta^+(U)) + x(\delta^-(U))) \label{eq:hoff4}\\
 &= 2\alpha x(\delta^+(U)) \label{eq:hoff5}\\
 &\leq 2\alpha u(\delta^+(U)) \label{eq:hoff6},
\end{align}
where \eqref{eq:hoff1} follows by the definition of $l$, \eqref{eq:hoff2} by the definition of $\vec{T}$, \eqref{eq:hoff3} by the thinness of $T$, \eqref{eq:hoff4} by the definition of $z$, \eqref{eq:hoff5} by flow conservation, and \eqref{eq:hoff6} by the definition of $u$.
This shows that condition (2) of theorem \ref{thm:hoffman_circulation} is satisfied, and therefore there exists a circulation $f$ as required.
Moreover, since $l$, and $u$ are integered-valued, $f$ can be chosen to be integer-valued.

We define the directed multigraph $\vec{H}$ containing an arc $(v,w)$ for every unit of flow on $(v,w)$ in $f$.
By flow conservation, $\vec{H}$ is Eulerian.
Moreover, $\vec{H}$ contains $\vec{T}$ as a subgraph.
Since $\vec{T}$ has $k$ weakly-connected components, it follows that $\vec{H}$ contains $k' \leq k$ strongly-connected components.
Therefore, $\vec{H}$ can be decomposed into $k'$ closed walks $C_1,\ldots,C_{k'}$ spanning $\dG$, as required.

It remains to bound the total cost of these walks.
%For any $i\in \{1,\ldots,k'\}$, suppose that $C_i$ visits the vertices $v_{i,0},\ldots,v_{i,k_i-1}$, in this order (and possibly with multiplicities).
We have
\begin{align*}
\sum_{i=1}^{k'} c(C_i) &= \sum_{a\in A} c(a) \cdot f(a)\\
 &\leq \sum_{a\in A} c(a) \cdot u(a)\\
 &= c(T) + 2\alpha \sum_{a\in A} c(a) \cdot x(a)\\
 &= (2\alpha + s) \sum_{a\in A} c(a) \cdot x(a),
\end{align*}
concluding the proof.
\end{proof}

\section{The algorithm}\label{sec:algorithm}

We are now ready to prove our main result.
%Let $(\dG,c)$ be an instance of ATSP, with $\eg(\dG)=g$, let $x$ be a feasible solution to the Held-Karp relaxation of $(\dG,c)$, and let $z$ be the symmetrization of $x$.
The first ingredient in our algorithm is a procedure which computes a $(O(1),O(1))$-thin forest $T$, having at most $g$ connected components.
This is done in Lemma \ref{lem:thin2}.
The second ingredient is a procedure which given this forest $T$, outputs a collection of at most $g$ walks that visit all vertices, and with  total cost at most a constant factor times the cost of the fractional solution $x$, i.e.~at most $O(1) \cdot \sum_{a\in A} c(a) \cdot x(a)$.
This is done in Lemma \ref{lem:circulation}.
Using these Lemmas as subroutines, we can obtain an approximation algorithm for ATSP.

\begin{proof}[Proof of Theorem \ref{thm:main}]
%Let $(\dG, c)$ be an instance of ATSP, where $\eg(G) = g$.
Let $\OPT\geq 0$ be the minimum cost of a solution to $(\dG, c)$.
We first compute an optimal solution $x$ to the Held-Karp LP relaxation of $(\dG, c)$.
Using lemma \ref{lem:thin2}, we compute in polynomial time a $(O(1),O(1))$-thin (w.r.t.~$x$) spanning forest $T$ in $G$, such that $T$ has at most $g$ connected components.
Using lemma \ref{lem:circulation}, we compute in polynomial time a collection of $k$ closed walks $C_1,\ldots,C_k$, with $k\leq g$, such that every vertex in $G$ is visited by some walk, and the total cost of all walks is at most $O(\OPT)$.
For every $i\in \{1,\ldots,k\}$, pick a vertex $v_i$ visited by $C_i$.
We create a new instance $(\dG', c')$ of ATSP induced on $\{v_1,\ldots,v_k\}$.
More precisely, we define the directed graph $\dG'=(V',A')$, with $V' = \{v_1,\ldots,v_k\}$, and for any $u,v \in V'$, we have an arc $u \rightarrow v$ in $A'$, with $c'(u \rightarrow v) = \min_{P} c(P)$, where $P$ ranges over all directed paths form $u$ to $v$ in $\dG$.
Clearly, the instance $(\dG', c')$ has a solution of cost $\OPT' \leq \OPT$.
Using the $f(n)$-approximation algorithm for $n$-vertex graphs, we find a closed tour $C$ in $\dG'$, of cost at most $\OPT \cdot f(k)$.
By composing $C$ with the closed walks $C_1,\ldots,C_k$, and shortcutting appropriately (as in \cite{DBLP:journals/networks/FriezeGM82}), we obtain a solution for $(\dG, c)$, of total cost $O(\OPT) + \OPT \cdot f(k) = O(\OPT \cdot f(k))$, as required.
\end{proof}

We can similarly show that in time exponential in the genus, we can obtain a constant-factor approximation.
This is shown in the following.

\begin{proof}[Proof of Theorem \ref{thm:FPT}]
We proceed exactly as in the proof of Theorem \ref{thm:main}.
The only difference is that once we compute the graph $\dG'$, instead of running an approximation algorithm on the instance defined by $\dG'$, we now find an \emph{exact} solution.
This can be done in time $2^{O(|V'|)} \cdot n^{O(1)} = 2^{O(g)} \cdot n^{O(1)}$, via a (folklore) dynamic program.
By composing this optimum solution with the $k$ closed walks $C_1,\ldots,C_k$, we obtain a solution for the original instance, of total cost at most $O(\OPT) + \OPT = O(\OPT)$, as required.
\end{proof}

\bibliography{bibfile}
\bibliographystyle{newuser}

\end{document}